\documentclass[a4paper, 12pt]{article}

\usepackage{latexsym}

\usepackage{amsmath}
\usepackage{epsf}
\usepackage[latin1]{inputenc}
\usepackage{graphics}
\usepackage{enumerate}
\usepackage{calc}
\usepackage{color}
\usepackage{amssymb}
\usepackage[active]{srcltx}
\usepackage{pstricks}
\usepackage{epsfig}

\usepackage{bbm}
\usepackage{amsthm}
\usepackage{amscd}
\usepackage{esint}
\usepackage{authblk}

\newcommand{\N}{\mathbb{N}}

\renewcommand{\S}{\mathbb{S}}

\newcommand{\p}{\varphi}
\newcommand{\la}{\lambda}
\renewcommand{\a}{\alpha}
\renewcommand{\b}{\beta}

\def\d{\partial}

\newcommand{\h}[1]{\widehat{#1}}
\def\e{\varepsilon}

\newcommand{\R}{\mathbb{R}}

\newcommand{\C}{\mathbb{C}}
\def\NN{\mathbb{N}}

\newcommand{\A}{\mathfrak{A}}

\def\a{\alpha}
\def\b{\beta}

\renewcommand{\d}{\partial}

\newcommand{\zwo}[2]{\begin{pmatrix} {#1}\\{#2} \end{pmatrix}}

\renewcommand{\t}{\widetilde}
\renewcommand{\o}{\overline}
\renewcommand{\P}{\mathbb{P}}

\newtheorem{theorem}{Theorem}[section]
\newtheorem{lemma}[theorem]{Lemma}
\newtheorem{proposition}[theorem]{Proposition}
\newtheorem{corollary}[theorem]{Corollary}
\newtheorem{definition}[theorem]{Definition}

\newcommand{\D}{\nabla}
\renewcommand{\h}{\widehat}

\usepackage{pdfpages} 
\usepackage{pgfplots}
\pgfplotsset{compat=newest}

\begin{document}
\title{Deformation of Framed Curves}

\author{Peter Hornung\footnote{Fakult\"at Mathematik, TU Dresden, 01062 Dresden (Germany),
\tt{peter.hornung@tu-dresden.de}}}

\date{}

\maketitle

\begin{abstract}
We consider curves $\gamma : [0, 1]\to\R^3$ endowed with an adapted orthonormal
frame $r : [0, 1]\to SO(3)$. We are interested
in the cases where the frame is constrained, in the sense that one
of its `curvatures' (i.e., off-diagonal elements of $r'r^T$)
is prescribed. One example is the Fr\'enet frame.
\\
In order to deform such constrained framed curves without spoiling the constraint,
we proceed in two stages. First we deform the frame $r$ in a way
that is naturally compatible with the differential constraint,
by interpreting it in terms of parallel transport on the sphere.
This provides a link to the differential geometry of surfaces.
\\
The deformation of the base curve $\gamma$ is achieved in a second step, by means
of a suitable reparametrization of the frame.
\\
We illustrate this deformation procedure by providing some applications:
first, we characterize the boundary conditions that are accessible
without violating the constraint; second, we provide a short and intuitive proof of
a smooth approximation result for constrained framed curves. In
both examples we prescribe clamped ends.
Finally, we also apply these
ideas to elastic ribbons with nonzero width.
\end{abstract}

\section{Introduction}

Elastic inextensible rods can be modelled by framed curves,
i.e., pairs $(\gamma, r)$ consisting of a curve
$\gamma$ mapping some interval $I\subset\R$ into $\R^3$
and a frame field $r : I\to SO(3)$ which is adapted to $\gamma$
in the sense that $r^T e_1 = \gamma'$.
In the context of the Cosserat theory of rods, the rows
$r_2, r_3$ are called directors.
For such rod theories and their generalizations 
we refer, e.g., to the monographs \cite{Antman, Ciarlet2}.
A concise exposition of Cosserat theory can be found, e.g., in
\cite{Schuricht}. In this context
Sobolev spaces of curves arise naturally; cf. also
\cite{Maddocks84}. Dynamical problems involving Kirchhoff rods \cite{Kirchhoff} were studied, e.g.,
in \cite{LinSchwetlick, AcquaLinPozzi}.
Rods with fixed ends play a role in \cite{FonsecaAguiar}.
\\
Framed curves also arise naturally in differential geometry,
because there is a frame field -- the Darboux frame --
that is naturally associated with any curve $\gamma$ in an immersed surface.
For this reason, framed curves play a role
in asymptotic theories for nonlinearly elastic plates \cite{FJM, H-CPAM} and,
in particular, ribbons \cite{KirbyFried, FHMP-SIAM, DiasAudoly, ParoniTomassetti}
such as the developable M\"obius strip
\cite{Sadowsky2, Wunderlich, StarostinHeijden, BartelsHornung}.
In \cite{DiasAudoly} elastic ribbons are described by framed curves endowed with
a `material' frame instead of the Fr\'enet frame used, e.g., in \cite{Wunderlich, KirbyFried}.
For a general approach we refer to \cite{H-CPAM, H-ARMA2}.
\\
In these situations a differential constraint
of the form $a_{ij} := (r^Te_i)'\cdot r^T e_j = 0$ arises naturally.
More precisely, in \cite{Wunderlich, StarostinHeijden, KirbyFried, RandrupRogen} the Fr\'enet frame 
is used to model a narrow ribbon; it is characterized
by the differential
constraint $a_{13}\equiv 0$. In \cite{DiasAudoly, FHMP-SIAM}
plates and ribbons are described by
framed curves satisfying the differential constraint $a_{12}\equiv 0$.
They can also described by framed curves arising from lines of curvature
\cite{H-CPAM, H-ARMA2}, which therefore satisfy the differential constraint
$a_{23}\equiv 0$ of vanishing torsion; this constraint also characterizes
the relatively parallel frame in \cite{Bishop}.
\\
Constrained framed curves also arise in DNA models
\cite{ChouaiebMaddocks, ChouaiebGorielyMaddocks}.
They have also been studied for their own sake,
see e.g. \cite{HondaTakahashi, YilmazTurgut, Rogen99, Scherrer, Solomon, Hwang}.
For a discussion of the various ways to frame a given
space curve we refer to \cite{Bishop, Silva}.
\\
For applications it is therefore useful to understand the possible shapes that
can be attained by constrained framed curves and 
to be able to deform framed curves
while preserving differential constraints of the form $a_{ij}\equiv 0$.
Such results have proven to be useful in numerical analysis \cite{BartelsReiter, Bartels2020}
as well as the calculus of variations \cite{H-CPAM, FHMP-SIAM, FHMP-clamped}.
\\
In the present article we address the same basic problem as in \cite{H-framed},
but with a very different approach, which is more
geometric and which leads to simple and intuitive constructions.
\\

More precisely, in the present article the deformation process is carried out in two successive stages.
In the first stage we modify the frame $r : I\to SO(3)$ while preserving
the differential constraint and possibly the value of $r$ at its endpoints.
In contrast to earlier work, this is not achieved by explicitly
modifying the curvatures $a_{ij}$ within some suitable Banach manifold
and subsequently solving an ordinary
differential equation. Instead, we interpret constraints of the form $a_{ij}\equiv 0$
geometrically in terms of parallel transport on the sphere.
This simple geometric viewpoint is explained in more detail in Section \ref{Parallel Transport}.
It allows us to work directly with the frames, which is
more intuitive than working with the curvatures.
Moreover, it provides a link to differential geometry
and therefore allows us to exploit some tools and notions from the 
differential geometry of surfaces. 
As a result, it also leads to a simple answer to a question
left open in \cite{H-framed}, cf. Section \ref{Accessible} below.
\\
In the second stage we reparametrize the frame in a suitable way.
This stage, described in Section \ref{Convex Integration},
is somewhat reminiscent of one-dimensional convex integration \cite{Gromov}.
\\

This article is organized as follows.
In Section \ref{Stages} we
we briefly introduce some basic notions and notation regarding framed curves,
and then we proceed to describe in some detail the two stages outlined earlier.
In Section \ref{Accessible} and Section \ref{Density} we provide two
simple applications of this method:
\\
In Section \ref{Accessible} we characterize
the possible clamped boundary conditions on a framed curve
that are compatible with the differential constraint. 
We show that a single constraint of the form $a_{ij}\equiv 0$ does
not restrict the boundary conditions that can be accessed by a framed curve,
therefore generalizing a corresponding local result in \cite{H-framed}.
\\
In Section \ref{Density}
we give a self-contained and short proof of one of the main results from
\cite{H-framed, BartelsReiter},
namely a smooth approximation result for framed
curves with prescribed endpoints. The proof given here differs fundamentally
from the proofs in those papers.
\\
In Section \ref{Ribbons} we 
consider ribbons with finite (as opposed to infinitesimal) width.
Such objects are more rigid than (infinitesimally narrow)
framed curves, because the admissible deformations are
isometric immersions of genuinely two-dimensional reference domains.
This requires geometric constructions which, while still elementary,
are less straightfoward than
those in the earlier sections. They combine
the approach to framed curves from earlier sections on one hand
with results about isometric immersions on the other.
Regarding the latter, we need
some observations about the Gauss (spherical image) map of isometric immersions;
in particular we need a good control on
its geodesic curvature (suitably defined) along regular curves.


\paragraph{Notation.}
We denote by $SO(n)$ the
set of rotation matrices. 
We identify $SO(2)$ with $\S^1\subset\C$ in the usual way.
The letter $I$ will denote both the identity matrix
and the open interval $(0, 1)$.
We denote by $(e_1, e_2, e_3)$ the canonical basis of $\R^3$
and by $B_r(a)$ the open ball of radius $r$ centered at $a$.

\section{Parallel Transport and Na\"ive Convex Integration}\label{Stages}

\subsection{Constrained Framed Curves}

In this article,
a framed curve $(\gamma, r)$ consists of an absolutely continuous frame field $r : I\to SO(3)$,
whose rows will henceforth be denoted by 
$$
r_i = r^Te_i,
$$
and of a curve $\gamma : I\to\R^3$ satisfying $\gamma' = r_1$.
The frame $r$ satisfies an ordinary differential equation of the form
\begin{equation}
\label{spelled}
\begin{pmatrix}
r_1
\\
r_2
\\
r_3
\end{pmatrix}'
=
\begin{pmatrix}
0 & a_{12} & a_{13}
\\
- a_{12} & 0 & a_{23}
\\
- a_{13} & - a_{23} & 0
\end{pmatrix}
\begin{pmatrix}
r_1
\\
r_2
\\
r_3
\end{pmatrix},
\end{equation}
for suitable `curvatures' $a_{ij} : I\to\R$.
For instance, if $\gamma$ is a curve in a surface and $r_3$ denotes the normal
to the surface along $\gamma$, then $r$ is the so-called Darboux frame.
In this case, $a_{13}$ is the normal curvature of $\gamma$, whereas
$a_{12}$ is its geodesic curvature and $a_{23}$ is its geodesic torsion.
\\
As in \cite{H-framed}, let $k$, $\ell\in \{1, 2, 3\}$ be unequal
and set
\begin{equation}
\label{aallg}
\mathfrak{A} = 
\left\{
\begin{pmatrix}
0 & a_{12} & a_{13}
\\
- a_{12} & 0 & a_{23}
\\
- a_{13} & - a_{23} & 0
\end{pmatrix}\in\R^{3\times 3} : a_{k\ell} = 0
\right\}.
\end{equation}
This gives rise to three different two dimensional subspaces $\A\subset\R^{3\times 3}$,
depending on the choice of $k$, $\ell$.
(The cases considered in \cite{H-framed} when $\A$ has one or three dimensions
are not relevant here.) From now on, $\A$ always is one of the sets \eqref{aallg}.

\subsection{
First Stage: Parallel Transport on the Sphere}\label{Parallel Transport}

In this short section we describe the viewpoint adopted in the
first stage of the deformation process.
\\
For a given absolutely continuous curve $\b : [a,b]\to\S^2$
let us denote by $\Pi_{\b}$ the linear isomorphism from the tangent
space $T_{\b(a)}\S^2$ to $\S^2$ at the point $\b(a)$ to the
tangent space $T_{\b(b)}\S^2$
that maps a tangent vector $v\in T_{\b(a)}\S^2$ to the tangent vector 
$\Pi_{\b}v\in T_{\b(b)}\S^2$
obtained by parallel transporting $v$ along $\b$ from $\b(a)$
to its endpoint $\b(b)$.

\begin{proposition}\label{propa}
Let $(i,k,\ell)$ be a permutation of $(1,2,3)$,
let $r : I\to SO(3)$ be absolutely continuous and assume
that $r_k'\cdot r_{\ell} \equiv 0$. 
Define $\b : I\to\S^2$ by $\b = r_i$. Then we have
\begin{equation}
\label{patr-1}
r_j(t) = \Pi_{\b|_{(0, t)}}r_j(0)\mbox{ for all }t\in I
\end{equation}
for $j = k$, $\ell$. 
\end{proposition}
\begin{proof}
Since $r_k$ is a unit vector field, we have $r_k'\cdot r_k \equiv 0$.
Hence the constraint $r_k'\cdot r_{\ell} \equiv 0$
implies that $r_k'$ is normal to $\S^2$ along $\b$.
Hence $r_k$ is parallel transported along the curve $\b$.
The same is true for $r_{\ell}$.
\end{proof}

In view of Proposition \ref{propa}, 
prescribing the end value $r(1)$ of the frame $r$ is equivalent
to prescribing the end value $\b(1)$ of the curve $\b : I\to\S^2$
as well as the associated parallel transport map $\Pi_{\b} : T_{\b(0)}\S^2\to T_{\b(1)}\S^2$.
\subsection{Second Stage: Na\"ive Convex Integration}\label{Convex Integration}

A reparametrization of the curve $\b$ clearly does not change its
endpoints, nor does it affect the parallel transport map $\Pi_{\b}$.
It does, however, affect the base curve $\gamma : t\mapsto\int_0^t r_1$.
This naturally leads us to consider reparametrizations of $\b$ (or, equivalently, of the frame $r$)
as a separate degree of freedom. This degree of freedom leads to different
curves $\gamma$ corresponding to the same trace $r(I)\subset SO(3)$.
\\
Let us be more precise. For a given continuous frame field $r : I\to SO(3)$ we
define reparametrized frame fields as follows.
Let $\eta : I\to\R$ be integrable, bounded from below by a positive constant and such that $\int_0^1\eta = 1$.
For each such $\eta$ we define the reparametrized frame field $\t r^{(\eta)} : I\to SO(3)$ by setting
$$
\t r^{(\eta)}\left(\int_0^t\eta\right) = r(t)\mbox{ for all }t\in I.
$$
Clearly, if $r_1$ is constant,
then so is $\t r_1^{(\eta)}$, for whichever choice of $\eta$.
In this case 
$\int_0^1 \t r_1^{(\eta)} = \int_0^1 r_1$
for all $\eta$ as above. The generic case is covered by the following proposition.
It is related to one-dimensional convex integration \cite{Gromov}, but clearly
the situation considered here is very basic.

\begin{proposition}\label{convex}
Let $r\in C^0(\o I, SO(3))$ and
let $\o\gamma$ be a point in the interior of the convex hull
of $r_1(\o I)\subset\R^3$. Then there exists an everywhere
positive function $\eta\in C^{\infty}(\o I)$
with $\int_I \eta = 1$ such that $\int_0^1 \t r^{(\eta)}_1 = \o\gamma$.
\end{proposition}
\begin{proof}
There exist $N\in\N$, $\e > 0$ and $0 \leq t_1 < \cdots < t_N \leq 1$ such that $B_{\e}(\o\gamma)$
is contained in the convex hull of
$\{r_1(t_1), ..., r_1(t_N)\}$. Let $\delta$ be as furnished by Lemma \ref{le-convex} below, applied
with $a_i = r_1(t_i)$ and $x = \o\gamma$.
\\
For all $i = 1, ..., N$ we can find a function $\eta_i\in C^{\infty}(\o I)$ with $\int_I\eta_i = 1$
which is positive everywhere on $\o I$,
and for which 
\begin{equation}
\label{definitionvonai}
\t a_i = \int_I \eta_i r_1
\end{equation}
is contained in $B_{\delta}(r_1(t_i))$.
Such $\eta_i$ exist because $r_1$ is continuous and for each $i$ there exists a sequence of 
positive functions in $C^{\infty}(\o I)$ with unit mass converging weakly-$*$
as measures to the Dirac measure concentrated at $t_i$.
\\
By Lemma \ref{le-convex} the point $\o\gamma$ is contained in the convex hull of the $\t a_i$.
Hence by \cite[Chapter I, Theorem 2.3]{Rockafellar} there exist nonnegative numbers $\la_1$, ..., $\la_N$
with 
\begin{equation}
\label{convex-2}
\sum_{i = 1}^N\la_i = 1 
\end{equation}
and such that
\begin{equation}
\label{convex-1}
\o\gamma = \sum_{i = 1}^N\la_i\t a_i.
\end{equation}
Set $\eta = \sum_{i = 1}^N \la_i\eta_i$. Then $\eta\in C^{\infty}(\o I)$ is positive
because all $\eta_i$ are positive and at least one $\la_i$ is nonzero. Moreover, $\int_I\eta = 1$
due to \eqref{convex-2}. By the definition \eqref{definitionvonai} of $\t a_i$,
$$
\int_0^1 \t r^{(\eta)}_1
= \int_0^1 \eta r_1
= \sum_{i = 1}^N\la_i\t a_i.
$$
By \eqref{convex-1} the right-hand side agrees with $\o\gamma$.
\end{proof}

In the proof of Proposition \ref{convex} we used the following lemma,
the proof of which is left to the reader.

\begin{lemma}\label{le-convex}
Let $N\geq 4$, let $x$, $a_1$, ..., $a_N\in\R^3$ and let $\e > 0$. Then there exists $\delta > 0$
such that the following is true: If $B_{\e}(x)$ is contained in the convex hull of $\{a_1, ..., a_N\}$
and $\t a_1, ..., \t a_N\in\R^3$ are such that $|\t a_i - a_i| \leq\delta$ for $i = 1, ..., N$,
then $x$ is contained in the convex hull of $\{\t a_1, ..., \t a_N\}$.
\end{lemma}

\subsection{Frames with Prescribed Curvature}\label{Prescribed}

In Section \ref{Parallel Transport}
we only focussed on homogeneous constraints of the form $a_{ij}\equiv 0$.
As an immediate consequence, one can also handle
inhomogeneous constraints of the form $a_{ij} = \kappa$ with
a prescribed function $\kappa\in L^1(I)$. This is achieved 
easily.
In order to be specific, we will only
consider the constraint $a_{12} = \kappa$. The results of this section
will not be needed until Section \ref{Ribbons}.

\begin{lemma}\label{reggi}
Let $p\in [1, \infty]$. If $\b\in W^{1,p}(I, \S^2)$ then
$t\mapsto \Pi_{\b|_{(0, t)}}v$ is in $W^{1,p}(I)$,
for all $v\in T_{\b(0)}\S^2$.
\end{lemma}
\begin{proof}
Since $\b$ is continuous, we may assume that $\b(I)$
is contained in the upper hemisphere.
Let $v\in T_{\b(0)}\S^2$ and set $V(t) = \Pi_{\b|_{(0, t)}}v$.
Then
$
V' = \Gamma(V, \b'),
$
where $\Gamma(V, \b')$ denotes a suitable contraction
with the (smooth) Christoffel symbols of $\S^2$ in the chosen
coordinates. Since $|V(t)| = |v|$ for almost every $t\in I$, we conclude
that indeed $V'\in L^p$.
\end{proof}

\begin{proposition}
\label{gauge}
Let $\kappa\in L^1(I)$
and
let $r\in L^{\infty}(I, SO(3))$ be such that $r_3\in W^{1,1}(I)$.
Then the following are equivalent:
\begin{enumerate}[(i)]
\item \label{gauge-i} For all $t\in I$ we have
\begin{equation}
\begin{split}
\label{gauge-4}
r_1(t) 
&=
\cos\left(\int_0^t\kappa\right) \Pi_{r_3|_{(0,t)}}r_1(0) 
+ \sin\left(\int_0^t\kappa\right) \Pi_{r_3|_{(0,t)}}r_2(0).
\end{split}
\end{equation}
\item \label{gauge-ii}
\label{darb-rem-b}
$r\in W^{1,1}(I)$ and $r_1'\cdot r_2 = \kappa$ almost everywhere.
\item \label{gauge-iii}
\label{darb-rem-c}
$r\in W^{1,1}(I)$ and there exist $\tau$, $\mu\in L^1(I)$ such that
\begin{equation}
\label{darb-cor-1}
r' = 
\begin{pmatrix}
0 & \kappa & \mu
\\
-\kappa & 0 & \tau
\\
-\mu & -\tau & 0
\end{pmatrix}
r
\end{equation}
\end{enumerate}
\end{proposition}
\begin{proof}
The equivalence of \eqref{darb-rem-c}
and \eqref{darb-rem-b} is clear.
\\
Notice that if \eqref{gauge-i} is satisfied,
then Lemma \ref{reggi} implies that $r\in W^{1,1}(I)$.
So in order to prove the equivalence of \eqref{gauge-i}
and \eqref{gauge-ii} we must show the following: if
$r$, $\t r\in W^{1,1}(I, SO(3))$ satisfy $r_3 = \t r_3$
and (with $K(t) = \int_0^t\kappa$)
\begin{equation}
\label{gauge-01}
r_1 = \t r_1\cos K + \t r_2\sin K,
\end{equation}
then 
\begin{equation}
\label{gauge-1a}
\t r_1(t) = \Pi_{r_3|_{(0, t)}} r_1(0)
\end{equation}
if and only if $r_1'\cdot r_2 = \kappa$.
\\
To verify this it is enough to take derivatives on both sides of \eqref{gauge-01}
to see that
$$
r_1'\cdot r_2 = \t r_1'\cdot\t r_2 + \kappa.
$$
So $r_1'\cdot r_2 = \kappa$ if and only if $\t r_1'\cdot\t r_2 = 0$,
which in turn is equivalent to \eqref{gauge-1a}.
\end{proof}

\paragraph{Remarks.}
\begin{enumerate}[(i)]
\item
In what follows, for $z\in\S^2$, an angle $\theta\in\R$ and $v\in T_z\S^2$,
we will use the notation
$$
R_z^{(\theta)}v = v\cos\theta + (z\times v)\sin\theta.
$$
This describes a counter-clockwise rotation by an angle $\theta$ of the tangent vector
$v$ within the tangent plane $T_z\S^2$.
\\
With this notation, \eqref{gauge-4} can be written as
$$
r_1(t) = R_{r_3(t)}^{\left( \int_0^t\kappa \right)}\Pi_{r_3|_{(0,t)}}r_1(0).
$$
\item Let $\a\in (\frac{1}{2},1)$. Due e.g. to results by Borisov,
the parallel transport
$\Pi_{\b}$ along a curve $\b\in C^{0, \a}(I, \S^2)$
is well-defined. Therefore, one can use Proposition \ref{gauge}
to make the following definition:
A frame $r\in L^{\infty}(I, SO(3))$ with $r_3\in C^{0, \a}(I)$
is said to have geodesic curvature $\kappa\in L^1(I)$ if
it satisfies \eqref{gauge-4}.
\end{enumerate}

\subsection{Curves with Bounded Geodesic Curvature}\label{Bdgeo}

In the following definition $\b$ need not be immersed.

\begin{definition}
\label{bdgeo}
An adapted frame for a curve $\b\in W^{1,1}(I, \S^2)$
is a map $r\in L^{\infty}(I, SO(3))$ satisfying
\begin{align}
\label{darb-1}
r_3 &= \b \mbox{ almost everywhere on }I
\\
\label{darb-2}
\beta'\times r_1 &= 0\mbox{ almost everywhere on }I.
\end{align}
\end{definition}

For $\b\in W^{1,1}(I, \S^2)$ define the open set
$$
C_{\b} = \{t\in I : \b\mbox{ is constant in a neighbourhood of }t\}.
$$
Let us denote the projective plane by $\P^2$.
For $v$, $w\in\S^2$ we will write
$$
v = w \mbox{ in }\P^2
$$
to mean that $v = w$ or $v = -w$.

\begin{lemma}\label{adaptedunique}
Let $\b\in W^{1,1}(I, \S^2)$ and let
$r\in L^{\infty}(I, SO(3))$
be an adapted frame for $\b$. Then the following are true:
\begin{enumerate}[(i)]
\item \label{adaptedunique-i} We have
$
\b' = (\b'\cdot r_1)\ r_1
$
almost everywhere on $I$ and
\begin{equation}
\label{adun-1}
r_1 = \frac{\b'}{|\b'|}\mbox{ in $\P^2$ almost everywhere on }\{\b'\neq 0\}.
\end{equation}
\item We have $r_3\in W^{1,1}(I)$ and
\begin{equation}
\label{adaptedlocal-torsion}
r_3'\cdot r_2 = 0\mbox{ almost everywhere on }I.
\end{equation}
\item If $r\in C^0(I, SO(3))$ and if $\t r\in C^0(I, SO(3))$ is another adapted frame, then
$$
\t r_1 = r_1\mbox{ in $\P^2$ on }I\setminus C_{\b}.
$$
\end{enumerate}
\end{lemma}
\begin{proof}
Equation \eqref{adun-1} is obvious.
Equation \eqref{adaptedlocal-torsion}
follows from \eqref{darb-1}, \eqref{darb-2} and the fact that $|\b|\equiv 1$.
\\
Now assume that $r$, $\t r$ are continuous adapted frames for $\b$.
Let $t_0\in I\setminus C_{\b}$. Then for all $\e > 0$ the set
$I_{\e} = (t_0-\e, t_0+\e)\cap\{\b' \neq 0\}$ has positive measure.
Hence by \eqref{adun-1} we have $r_1 = \t r_1$ in $\P^2$
almost everywhere on  $I_{\e}$, hence by the arbitrariness of $\e$
and by continuity $r_1(t_0) = \t r_1(t_0)$ in $\P^2$. 
\end{proof}

\begin{definition}\label{darb-def2}
A curve $\b\in W^{1,1}(I, \S^2)$ is said to have
geodesic curvature $\kappa\in L^1(I)$ if
there exists
an adapted frame $r\in W^{1,1}(I, SO(3))$ for $\b$ satisfying
\begin{equation}
\label{darb-3}
r_1'\cdot r_2 = \kappa \mbox{ almost everywhere on }I.
\end{equation}
\end{definition}

\paragraph{Remarks.}
\begin{enumerate}[(i)]
\item
An apparently more general definition would only require $r\in L^{\infty}(I)$
and replace \eqref{darb-3} by the condition that $r$ satisfy \eqref{gauge-4}.
However, when $\b\in W^{1,1}$ then it is equivalent to Definition
\ref{darb-def2}. This follows from 
Proposition \ref{gauge}.
\item In view of Lemma \ref{adaptedunique}
the geodesic curvature $\kappa$ is uniquely determined by $\b$
almost everywhere on $I\setminus C_{\b}$.
\end{enumerate}

The next lemma shows that having bounded geodesic curvature
is a local property.

\begin{lemma}\label{adaptedlocal}
Let $\b\in W^{1,1}(I, \S^2)$ and assume that for every $t_0\in \o I$ 
there exists $\e > 0$ such that the restriction
of $\b$ to $\o I\cap (t_0-\e, t_0+\e)$ has bounded geodesic curvature.
Then $\b$ has bounded geodesic curvature. 
\end{lemma}
\begin{proof}
We cover $\o I$ by finitely many relatively
open intervals $J_1$, ..., $J_N\subset\o I$
on each of which $\b$ has bounded geodesic curvature. So for each $i = 1, ..., N$
there is an adapted frame $r^{(i)}\in W^{1,1}(J_i, SO(3))$ for $\b$.
\\
By iteration we may assume without loss of generality
that $N = 2$ and that $0\in J_1$ and $1\in J_2$.
To construct an adapted frame defined on all of $J_1\cup J_2$,
we observe that two cases can occur:
\\
If 
$J_1\cap J_2\setminus C_{\b}$ is nonempty, then let $t_0$ be a point in this set
and define
$$
r =
\begin{cases}
r^{(1)}&\mbox{ on }(0, t_0]
\\
r^{(2)}&\mbox{ on }(t_0, 1).
\end{cases}
$$
The map $r$ is continuous because (after possibly replacing
$r^{(2)}$ by $-r^{(2)}$) we have $r^{(1)}(t_0) = r^{(2)}(t_0)$
by Lemma \ref{adaptedunique}. Hence $r\in W^{1,1}(I)$ and,
moreover, $r_1'\cdot r_2$ is bounded. So indeed $\b$ has bounded
geodesic curvature on $I$.
\\
If 
$
J_1\cap J_2\subset C_{\b}
$
then $J_1\cap J_2$ is contained in a maximal interval $(t_0, t_1)\subset C_{\b}$.
If $(t_0, t_1) = I$, then $\b$ is constant and there is nothing to prove.
Otherwise, we denote by $\t r_1 : [t_0, t_1]\to\S^2$ a smooth interpolation
from $r_1^{(1)}(t_0)$ to $r_1^{(2)}(t_1)$ within the unit circle in the
plane $T_{\b(t_0)}\S^2$.
Then we define $\t r : [t_0, t_1]\to SO(3)$
by setting $\t r_3 = \b$, and we define
$$
r =
\begin{cases}
r^{(1)} &\mbox{ on }[0, t_0]
\\
\t r &\mbox{ on }(t_0, t_1)
\\
r^{(2)} &\mbox{ on }[t_1, 1].
\end{cases}
$$
It is easy to see that $r\in W^{1,1}$ with $r_1'\cdot r_2\in L^{\infty}$,
because the same is true for $r^{(1)}$ and $r^{(2)}$.
\end{proof}

\section{Accessible Boundary Conditions}\label{Accessible}

A (midline of a) ribbon with clamped lateral boundaries
is described by a framed curve for which the initial value $(\gamma(0), r(0))$
and its end value $(\gamma(1), r(1))$ are prescribed.
However, a constraint of the form $a_{ij}\equiv 0$ clearly restricts the possible
shapes of framed curves $(\gamma, r)$. Therefore, it is a priori not obvious which 
clamped boundary conditions
can be achieved by such constrained framed curves.
\\
By applying a rigid motion we may assume without loss of generality
that $\gamma(0)$ is the origin and $r(0)$ is the identity matrix. Then
this amounts to the question which values $\gamma(1)$ and $r(1)$ can be attained
for the solution $r$ of $r' = Ar$ with inital value $r(0) = I$
and $\gamma(t) = \int_0^t r_1$, where $A$ is a map taking values
only in one of the three sets $\A$ defined by \eqref{aallg}.
\\
This question is answered in this section.
Exploiting the connection to parallel transport on the sphere elicited by the
first stage, the possible values of $r(1)$ will be readily seen to be arbitrary, 
because the holonomy group of $\S^2$ is $SO(2)$. The arbitrariness of $\gamma(1)$
then follows by carrying out the second stage.
\\
More precisely,
in this section we will prove the following result. It asserts that all boundary
conditions can be attained by a suitable choice of (even a smooth)
$A : I\to\A$.
\begin{proposition}\label{accessible}
For all $\o r\in SO(3)$ and all $\o\gamma\in B_1(0)$ there is
an $A\in C^{\infty}(\o I, \A)$ such that 
the solution $r$ of $r' = Ar$ with inital value $r(0) = I$
satisfies $r(1) = \o r$ and $\int_0^1 r_1 = \o\gamma$.
\end{proposition}

\paragraph{Remarks.}
\begin{enumerate}[(i)]
\item The proof will show that there is much freedom in the choice of $A$.
It could also be chosen to be piecewise constant.
\item In view of the second stage, we will see that it is enough to satisfy only the condition
$r(1) = \o r$. Below, we prove the existence of such an $r$ by means of the first stage,
but one could also apply, e.g., the techniques from \cite{H-framed} or use the properties of
the exponentional map for $SO(3)$.
\item An endpoint $\o\gamma$ with $|\o\gamma| = 1$ can clearly
only be achieved if $\o\gamma = e_1$, and in this case we must have $r_1\equiv e_1$.
\begin{itemize}
\item In the case $a_{23} \equiv 0$ this implies that $\b \equiv e_1$ and so necessarily $\o r = I$.
\item In the case $a_{13} \equiv 0$ we have $\b = r_2$, so
it implies that $\Pi_{\b|_{(0, t)}}e_1 = e_1$ for all $t$. This is the case if, and only
if, $r_2$ takes its values in the big circle $\S^2\cap \{e_1\}^{\perp}$.
(The fact that $r_2$ must be contained in a big circle of course follows as well from the observation that $r_2$ must take values in the orthogonal
complement of $e_1$.) So $\o r$ can be attained precisely if $\o r_2$ lies in the circle $\S^2\cap \{e_1\}^{\perp}$.
\item The case $a_{12}\equiv 0$ is similar to the case $a_{13}\equiv 0$.
\end{itemize}
\end{enumerate}

\subsection{Holonomy}

In order to prove Proposition \ref{accessible}, let us first show
that for all $\o r\in SO(3)$
there is some $A\in C^{\infty}(\o I, \A)$ such that the solution $r$ of $r' = Ar$ 
with $r(0) = I$ satisfies
$r(1) = \o r$.
This will follow from the connectedness of $\S^2$ and from the 
fact that the holonomy group of $\S^2$ is $SO(2)$.
In order to provide explicit constructions, 
we will now describe a simple way of introducing a loop into a given curve $\b$.
This construction allows us to preserve both the smoothness of $\b$
the fact that it is immersed. At the same time it allows to modify
the parallel transport along $\b$.

\subsubsection{A building block}\label{Loop}

We will define an immersed smooth parametrization $\h\b : \S^1\to\R^2$
of a square with rounded corners: let
$\t\rho : \R\to [1, \sqrt{2}]$ be the $\frac{\pi}{2}$-periodic
function determined by
$$
\t\rho(t) = \frac{1}{\cos t}\mbox{ for }t\in\left[-\frac{\pi}{4}, \frac{\pi}{4}\right].
$$
Now let $\rho : \R\to [1, \sqrt{2}]$ be a function
with the following properties:
\begin{itemize}
\item $\rho$ is even and $\frac{\pi}{2}$-periodic.
\item $\rho = \t\rho$ on $(-\pi/8, \pi/8)$.
\item $\rho'\in C^{\infty}_0(-\pi/4, \pi/4)$.
\item $\rho'\geq 0$ on $(0, \pi/4)$.
\end{itemize}
So $\rho$ is just a suitable smooth truncation of $\t\rho$. Now define the immersed smooth curve
$\Xi : \R\to\R^2$ by setting
$$
\Xi(t) = \rho(t) e^{it}.
$$
It is clearly $2\pi$-periodic and it parametrizes the boundary of
a square with rounded off corners, whose perimeter we denote by $\h\ell = \int_0^{2\pi}|\Xi'|$.
\\
We define $\h\b : \R\to\R^2$ by setting
$$
\h\b\left(\int_0^t |\Xi'| \right) = \Xi\left(t - \frac{\pi}{2}\right).
$$
Then $|\h\b'|\equiv 1$ and $\h\b(0) = -e_2$, and $\h\b$ is $\h\ell$-periodic.

\subsubsection{Concatenation of Continuous Curves}
\label{Concatenation}

In what follows we will consider reparametrized
concatenations of continuous curves
$$
\b^{(i)} : [t_0^{(i)}, t_1^{(i)}]\to X\mbox{ for }i = 1, ..., N
$$
into some manifold $X$ satisfying
$$
\b^{(i)}(t_1^{(i)}) = \b^{(i+1)}\left(t_0^{(i + 1)}\right)
\mbox{ for }i = 1, ..., N-1.
$$
The reparametrized concatenation
$
\b = \b^{(N)}\bullet \cdots \bullet \b^{(1)}
$
is the continuous curve $I\to X$ obtained by concatenation
of the $\b^{(i)}$ and subsequent reparametrization by a
constant factor.
\\
To make this explicit let us first consider the case $N=2$. Define
$\ell^{(i)} = t_1^{(i)} - t_0^{(i)}$ and $\ell = \ell^{(1)} + \ell^{(2)}$
as well as 
$$
I^{(1)} = \big(0, \frac{\ell^{(1)}}{\ell}\big]
\mbox{ and }
I^{(2)} = \left(\frac{\ell^{(1)}}{\ell}, 1\right).
$$
The concatenation
$
\b = \b^{(2)}\bullet \b^{(1)}
$
is the curve $\b : I\to X$ defined by
$$
\b(t) =
\begin{cases}
\b^{(1)}\left(\ell t + t_0^{(1)}\right)
&\mbox{ if }t\in I^{(1)}
\\
\b^{(2)}\left(\ell t - \ell^{(1)} + t_0^{(2)}\right)
&\mbox{ if }t\in I^{(2)}.
\end{cases}
$$
For $N > 2$ we define inductively
$$
\b^{(N)} \bullet \cdots \bullet \b^{(1)} =
\b^{(N)} \bullet \left(\b^{(N-1)} \bullet \cdots \bullet \b^{(1)}\right).
$$
In what follows we will also use the notation $I^{(i)}$ introduced
here (with obvious labels) to refer to the subintervals of $I$ 
where $\b$ agrees with $\b^{(i)}$ up to a reparametrization.

\subsubsection{Modifying the Parallel Transport}

The following lemma uses a straightforward way to modify the parallel transport
along a curve by adding a loop.

\begin{lemma}\label{gaubo}
Let $\b\in C^{\infty}(\o I, \S^2)$ be an immersed curve, let $\theta\in\R$
and let $V\subset\S^2$ be a neighbourhood of $\b(\o I)$. Then
there exists an immersed curve $\t\b\in C^{\infty}(\o I, V)$ which agrees with $\b$
in a neighbourhood of $\d I$ and which satisfies
$$
\Pi_{\t\b} = R_{\b(1)}^{(\theta)}\circ\Pi_{\b}.
$$
\end{lemma}
\begin{proof}
After possibly restricting $\b$ to a subinterval,
we may assume that $\b(\o I)$ is contained in an open subset $U\subset\S^2$ and that there
exists a bounded $C^{\infty}$ diffeomorphism $\Phi : U\to\Phi(U)\subset \R^2$
such that 
$$
\Phi\circ\b = \b^{(1)}\mbox{ on }I,
$$
where $\b^{(1)} : I\to\R^2$ is defined by
$$
\b^{(1)}(t) = t e_1.
$$
Let $\delta > 0$, $k\in\N$ and let $\h\b$ be the building block introduced in Section
\ref{Loop}.
Let $\h\ell$ be as in Section \ref{Loop} and
define the immersed smooth curve $\h\b_{\delta, k} : [0, k\delta\h\ell] \to\R^2$
by setting 
$$
\h\b_{\delta, k}(t) = \delta\h\b\left(\frac{t}{\delta}\right) + \delta e_2 + \frac{1}{2} e_1.
$$
This is a $k$-fold parametrization of the boundary of a smoothened version of 
the square 
$$
(1/2-\delta, 1/2+ \delta)\times (0, 2\delta)\subset\R^2;
$$
we denote the `square' with rounded off corners
by $S_{\delta}$. The perimeter of $S_{\delta}$
is $\delta\h\ell$.
\\
We define $\o\b : I\to\R^2$ by
\begin{equation}
\label{gaubo-1}
\o\b = \b^{(1)}|_{(\frac{1}{2},1)}\bullet\h\b_{\delta,k}\bullet\b^{(1)}|_{(0, \frac{1}{2})}.
\end{equation}
By choosing $\delta$ small enough we can ensure that the trace of $\o\b$ is contained in $\Phi(U\cap V)$.
Notice also that $\o\b$ is $C^{\infty}$ up to the boundary of $I$
because so is its derivative, since 
$\h\b_{\delta, k}' = e_1$ in a neighbourhood of $0$.
\\
Since, up to a smooth reparametrization, $\o\b$ agrees with $\Phi\circ\b$
near the boundary of its interval of definition,
there is an immersed $C^{\infty}$ curve $\t\b : \o I\to\S^2$, obtained as a suitable
smooth reparametrization of $\Phi^{-1}\circ\o\b$, which agrees with $\b$ near $\d I$.
In particular,
$$
\t\b(0) = \b(0)\mbox{ and }\t\b(1) = \b(1).
$$
Denoting by $\t s_{\delta}$ the area of $\Phi^{-1}(S_{\delta})$,
we have
$$
\Pi_{\t\b} = R^{(k\t s_{\delta})}_{\b(1)}\circ\Pi_{\b},
$$
by Corollary \ref{corle} below.
Choosing $\delta$ and $k\in\N$ suitably, the rotation angle $k\t s_{\delta}$
can be arranged to take any value. 
\end{proof}

\subsection{Proof of Proposition \ref{accessible}}

Set $\b^{(0)} = \b(0)$.
Since $\o\gamma$ is contained in the interior of the convex hull of $\S^2\subset\R^3$,
there exists an $N\in\N$ and $\o r_1^{(1)}, ..., \o r_1^{(N)}\in\S^2$ such that
$\o\gamma$ lies in the interior of the convex hull of $\{\o r_1^{(1)}, ..., \o r_1^{(N)}\}$.
Set $\o r^{(0)}_1 = e_1$ and $\o r_1^{(N+1)} = \o r_1$. Choose a partition $0 = t_0 < t_1 < \cdots < t_N < t_{N+1} = 1$ of $I$.
\\
Let us first consider the constraint $a_{12} = 0$.
Choose, for example, a constant speed geodesic $\t\b : \o I\to\S^2$ with endpoints $\t\b(0) = e_3$
and $\t\b(1) = \o r_3$.
By Lemma \ref{gaubo} we can replace $\t\b$ on each interval $[t_i, t_{i+1}]$
with a curve $\b\in C^{\infty}([t_i, t_{i+1}])$ in such a way that
$$
\Pi_{\b|_{(t_i, t_{i+1})}}\o r_1^{(i)} = \o r_1^{(i+1)}\mbox{ for all }i = 0, ..., N.
$$
By Lemma \ref{gaubo}
we also have $\b = \t\b$ near each $t_i$. Hence $\b$ is 
smooth on all of $[0, 1]$.
Define the frame $r : I\to SO(3)$ by
setting $r_3 = \b$ and, for $j = 1, 2$,
\begin{equation}\label{indu-1}
r_j(t) = \Pi_{\b|_{(0, t)}}e_j\mbox{ for all }t\in I.
\end{equation}
Then $r$ is smooth and satisfies $r_1'\cdot r_2 \equiv 0$; moreover, $r(1) = \o r$ and 
$r_1(t_i) = \o r_1^{(i)}$ for $i = 0, ..., N$. 
\\
In particular, $\o\gamma$ is contained in the interior of the convex 
hull of $r_1(I)$. So by Proposition \ref{convex}, after reparametrizing
$r$ smoothly, we can arrange that $\o\gamma = \int_0^1 r_1$.
\\
The case $a_{13} = 0$ is similar to the case $a_{12} = 0$, so
it remains to consider the case $a_{23} = 0$. Choose a smooth curve $\t\b : \o I\to\S^2$ 
with constant speed which satisfies $\t\b(t_i) = \o r_1^{(i)}$ for all $i = 0, ..., N+1$.
By Lemma \ref{gaubo} there exists a smooth curve $\b : I\to\S^2$ whose trace agrees
with $\t\b$ except for an additional small loop
and such that, moreover,
the frame $r : I\to SO(3)$ defined by
setting $r_1 = \b$ and by \eqref{indu-1} for $j = 2, 3$
satisfies $r(1) = \o r$.
Since, in particular, $\t\b(I)\subset \b(I)$ and $r_1 = \b$, we see that
$\o\gamma$ is contained in the interior of the convex hull
of $r_1(\o I)$. Hence we can use Proposition \ref{convex} as before.

\section{Smooth Approximation of Framed Curves}\label{Density}

As another application of the two stages, we will now prove the following result.

\begin{theorem}\label{sden}
For all $A\in L^2(I, \A)$
there exist $A_n\in C^{\infty}(\o I, \A)$ such that
$A_n\to A$ strongly in $L^2(I, \R^{3\times 3})$
and, moreover, the solutions
$r$, $r^{(n)}$ of $r' = Ar$ and
of $(r^{(n)})' = A_nr^{(n)}$ with $r(0) = r^{(n)}(0) = I$
satisfy
$$
r^{(n)}(1) = r(1)\mbox{ and }\int_0^1 r^{(n)}_1 = 
\int_0^1 r_1
\mbox{ for all }n\in\NN.
$$
\end{theorem}

\paragraph{Remarks.}
\begin{enumerate}[(i)]
\item We will obtain Theorem \ref{sden} as an immediate 
corollary of Proposition \ref{approx} below, which is the main
result of this section.
\item A more detailed result was obtained in \cite{H-framed} by a completely different approach;
compare also \cite{BartelsReiter}.
However, the two viewpoints are not entirely unrelated. The 
nonlocal constraint on the $a_{ij}$ in \cite{H-framed}
appears here in the guise of the (trivial) Lemma \ref{simpl} below.
\end{enumerate}

\subsection{Smooth Approximation Preserving Parallel Transport}

Using the viewpoint of Section \ref{Parallel Transport}, we see
that the essential point in the proof of Theorem \ref{sden}
will be to smoothly approximate a given $W^{1,2}$ curve on $\S^2$
while preserving its endpoint and its parallel transport map. The
correct endpoint for $\gamma$ is then achieved by carrying out stage two.

\begin{proposition}\label{approx}
Let $\b\in W^{1,2}(I, \S^2)$. Then there exist $\b_n\in C^{\infty}(\o I, \S^2)$
such that $\b_n$ converges to $\b$ strongly in $W^{1,2}$ and, moreover,
for all $n\in\N$ we have
\begin{enumerate}[(i)]
\item \label{approx-i}
$\b_n$ is constant near $\d I$,
\item \label{approx-ii}
$\b_n(0) = \b(0)$ and $\b_n(1) = \b(1)$,
\item \label{approx-iii}
$\Pi_{\b_n} = \Pi_{\b}$.
\end{enumerate}
\end{proposition}

To prove Proposition \ref{approx}, we recall
a well-known geometric fact.

\begin{lemma}\label{ptle}
Let $U\subset\S^2$ be open and let $(E_1, E_2)$ be a 
smooth tangent orthonormal frame field on $\o U$, and denote
by $\omega = E_2\cdot DE_1$ the corresponding connection form. Let
$\b$, $\t\b : I\to U$ be absolutely continuous and such that
$\b(0) = \t\b(0)$ and $\b(1) = \t\b(1)$.
Then
$\Pi_{\b} = \Pi_{\t\b}$ if and only if
$
\int_{\b}\omega = \int_{\t\b}\omega.
$
\end{lemma}
\begin{proof}
Let $V_0$ be a unit vector in the tangent space to $\S^2$ at $\b(0)$ and denote
by $V(t) = \Pi_{\b|_{(0, t)}}V_0$ its parallel transport along $\b$. Since $V$ is absolutely continuous, 
there exists an absolutely continuous function $\p : I\to\R$ such that
$$
V(t) = E_1(\b(t))\cos\p(t) + E_2(\b(t))\sin\p(t)\mbox{ for all }t\in I.
$$
We have
$$
\d_t V = \p'\cdot (-\sin\p E_1(\b) + \cos\p E_2(\b)) + D_{\b'}E_1(\b)\cos\p + D_{\b'}E_2(\b)\sin\p.
$$
Since $V$ is parallel along $\b$, we have $E_1\cdot \d_t V = E_2\cdot\d_t V = 0$.
We conclude that
$$
\omega(\b)(\b') + \p' = 0\mbox{ almost everywhere on }I,
$$
Integration over $I$ yields
\begin{equation}
\label{ptle-7}
\int_{\b}\omega = \p(0) - \p(1).
\end{equation}
The assertion follows from this formula and the same formula for $\t\b$.
\end{proof}

The following is a standard consequence of formula \eqref{ptle-7}.
\begin{corollary}\label{corle}
Let $S\subset\S^2$ be a simply connected smoothly bounded domain with area $\theta$ and
let $\b\in C^1(\S^1, \S^2)$ be a simple positively oriented parametrization of its
boundary $\d S$. Then $\Pi_{\b} : T_{\b(0)}\S^2\to T_{\b(0)}\S^2$ is given
by $\Pi_{\b} = R_{\b(0)}^{(\theta)}$.
\end{corollary}

We will also use the following simple lemma, the proof of which is left to the reader.

\begin{lemma}\label{simpl}
For every nonconstant $h\in L^{\infty}(I)$ 
and all $\delta\in\R$ there exists a $\h\mu\in C_0^{\infty}(I)$
such that $\int\h\mu = 0$ and $\int h\h\mu = \delta$, as well as
$$
\|\h\mu\|_{L^{\infty}}\leq |\delta|\cdot\frac{4\|h\|_{L^{\infty}(I)}}{\sigma^2_h},
$$
where $\sigma_h^2 = \int_I h^2 - \left(\int_I h\right)^2$.
\end{lemma}

\begin{proof}[Proof of Proposition \ref{approx}]
We may subdivide the interval $I$ into finitely many subintervals and restrict the construction to each subinterval.
In fact, due to \eqref{approx-i} and \eqref{approx-ii},
we will be able to smoothly glue together the smooth approximations obtained separately
on each subinterval. Therefore, since $\b$ is continuous,
we may assume without loss of generality that $\b(\o I)$ is contained
in the intersection of the (relatively open) upper and right hemispheres;
in particular, it does not contain the north pole.
\\
We parametrize the relevant portion of $\S^2$ by
\begin{align*}
\Psi : (0, 1)\times (-\pi, \pi) &\to\S^2
\\
(\rho, \p) &\mapsto (\rho\cos\p, \rho\sin\p, h(\rho)),
\end{align*}
where $h(\rho) = \sqrt{1 - \rho^2}$. We define the orthonormal frame field $(E_1, E_2)$ by
\begin{align*}
E_1(\rho, \p) &= (h(\rho)\cos\p, h(\rho)\sin\p, -\rho) 
\\
E_2(\rho, \p) &= (-\sin\p, \cos\p, 0).
\end{align*}
The connection form $\omega = E_2\cdot DE_1$ is
$
\omega = h(\rho)\ d\p.
$
In view of our hypotheses on the range of $\b$,
there exist $\rho$, $\p\in W^{1,2}(I)$ such that $\b(t) = \Psi(\rho(t), \p(t))$ for all $t\in I$.
\\
Let $\e > 0$ be small and let $\t\rho\in C^{\infty}(\o I)$ be such that $\|\t\rho - \rho\|_{W^{1,2}(I)} < \e$
and such that $\t\rho = \rho(0)$ near $0$ and $\t\rho = \rho(1)$ near $1$.
Let us assume for the moment that $\rho$ is not constant. 
Then, with the notation from Lemma \ref{simpl}, we have $\sigma_{h\circ\rho} > 0$ and
\begin{equation}
\label{approx-5}
\sigma_{h\circ\t\rho} \geq \frac{1}{2}\sigma_{h\circ\rho}\mbox{ for all $\e > 0$ small enough.}
\end{equation}
Set $\mu = \p'$.
Let $\o\mu\in C_0^{\infty}(I)$ be such that $\|\o\mu - \mu\|_{L^2}<\e$
and $\int_I\o\mu = \int_I\mu$.
Define
$$
\delta = \int h(\t\rho)\o\mu - \int h(\rho)\mu.
$$
By Lemma \ref{simpl} there exists $\h\mu\in C_0^{\infty}(I)$ satisfying
$\int_I \h\mu = 0$ and
$\int h(\t\rho)\h\mu = - \delta$ as well as 
\begin{equation}
\label{approx-4}
\|\h\mu\|_{L^{\infty}}\leq C|\delta|.
\end{equation}
Here, the constant $C$ can be chosen to be independent of $\e$
because by \eqref{approx-5} the variance $\sigma_{h\circ\t\rho}$ 
is bounded away from zero, uniformly in $\e$ for all $\e$ small enough.
\\
The function $\t\mu = \o\mu + \h\mu$ satisfies 
\begin{equation}
\label{approx-1}
\int\t\mu = \int\mu
\end{equation}
as well as
\begin{equation}
\label{approx-2}
\int h(\t\rho)\t\mu = \int h(\rho)\mu.
\end{equation}
In view of \eqref{approx-1} and \eqref{approx-2}
and since $\delta$ (and therefore $\h\mu$, by \eqref{approx-4}) can be made arbitrarily small by choosing $\e$ small enough,
the functions $\t\p(t) = \p(0) + \int_0^t\t\mu$ and $\t\rho$ define a smooth
curve $\t\b = \Psi(\t\rho, \t\p)$ with the desired properties
(here $\t\b$ plays the role of $\b_n$ in the statement).
In fact, \eqref{approx-1} implies that $\t\p(1) = \p(1)$;
since $\t\rho$ agrees with $\rho$ in $0$ and in $1$ this implies
that $\t\b$ and $\b$ have the same endpoints. 
The compact support of $\t\mu$ ensures that $\t\p$ is constant near $0$ and $1$,
and since the same is true for $\t\rho$, we see that \eqref{approx-i} is satisfied.
Finally, \eqref{approx-2} implies that
$
\int_{\t\b}\omega = \int_{\b}\omega
$
and therefore $\Pi_{\t\b} = \Pi_{\b}$, by Lemma \ref{ptle}.
\\
It remains to consider the degenerate case when $\rho$ is constant. Then
$h(\rho)$ is constant as well. We set $\t\rho = \rho$ and observe that in this
situation \eqref{approx-2} follows from \eqref{approx-1}.
Therefore, we can simply take $\h\mu = 0$ in the earlier argument. 
\end{proof}

\subsection{Proof of Theorem \ref{sden}}

We only consider the case $a_{12} = 0$; the others are similar.
Set $\b = r_3$ and let $\b_n\in C^{\infty}(\o I, \S^2)$
as provided by Proposition \ref{approx}. Define
$r^{(n)}\in C^{\infty}(\o I, SO(3))$
by setting $r^{(n)}_3 = \b_n$ and, for $i = 1, 2$,
$$
r^{(n)}_i(t) = \Pi_{\b_n|_{(0, t)}}r_i(0)\mbox{ for all }t\in I.
$$
Then $(r_1^{(n)})'\cdot r_2^{(n)}\equiv 0$,
and for $i = 1,2$ we have
$$
a_{i3}^{(n)} = - r_i^{(n)}\cdot (r_3^{(n)})'\to a_{i3}\mbox{ strongly in }L^2(I),
$$
because $r^{(n)}\to r$ uniformly and $\b_n'\to\b'$ strongly in $L^2$.

\section{Ribbons with Finite Width}\label{Ribbons}

The purpose of this section is
to prove the following version of Proposition \ref{accessible}
for finite ribbons, i.e., ribbons with nonzero width. Such finite ribbons
are more rigid than infinitesimal ribbons (i.e., framed curves).

\begin{theorem}\label{ribbon-accessible}
For all $\o r\in SO(3)$ and all $\o\gamma\in B_1(0)$ there is a constant $w > 0$
and a $C^{\infty}$ isometric immersion
$$
u : \o I\times [-w, w]\to\R^3
$$
satisfying $u(0, 0) = 0$
as well as
$
u(1, 0) = \o\gamma
$
and, for $i = 1, 2$,
\begin{align*}
\d_i u &= e_i\mbox{ on }\{0\}\times (-w,w)
\\
\d_i u &= \o r_i\mbox{ on }\{1\}\times (-w,w).
\end{align*}
\end{theorem}

The idea of the proof will be similar to that of Proposition \ref{accessible},
in that we combine parallel transport and with a suitable reparametrization of the frame.
However, here we need to satisfy several additional conditions.
Nevertheless, the construction only makes use of piecewise geodesics and
is elementary geometric.
So it contrasts sharply with the abstract approach in \cite{H-framed}.
\\
The proof of Theorem \ref{ribbon-accessible} will be obtained 
by combining facts about the Gauss map of isometric immersions
with a result involving only framed curves.

\subsection{Gauss Map of Isometric Immersions}

Let $S\subset\R^2$ be a bounded domain.
By the Gauss map of an immersion $u : S\to\R^3$
we mean the map from the reference domain $S\subset\R^2$ into $\S^2$
given by
$$
n(x) = \frac{\d_1 u(x)\times \d_2 u(x)}{|\d_1 u(x)\times \d_2 u(x)|}.
$$
If $u$ is an isometric immersion $S\to\R^3$ (where $S$ is
endowed with the standard flat metric),
i.e., $\d_i u\cdot\d_j u = \delta_{ij}$, then the denominator is always $1$.
For $W^{2,2}$ isometric immersions the usual formulae for smooth immersion remain
true, cf. \cite{FJM2}. In particular, we have that
\begin{equation}
\label{h2n}
\d_i\d_j u(x) = h_{ij}(x) n(x)\mbox{ for almost every }x\in S,
\end{equation}
where $h_{ij}$ are the components of the second fundamental form.
Now let $b : I\to S$ be a Lipschitz curve and set $r_i(t) = \d_i u(b(t))$
for $i = 1, 2$ and $r_3(t) = n(b(t))$. Then \eqref{h2n} shows that
$$
r_1'(t)\cdot r_2(t) = 0.
$$
Hence $r_{1}$ and $r_2$ are parallel transported along the
spherical curve $\b : I\to\S^2$ defined by $\b = r_3$.

\begin{proposition}\label{bdcurv}
Let $S\subset\R^2$ be a bounded domain, let $u\in W^{2,2}_{\delta}(S)$
and let $b : \o I\to S$ be a $W^{2,\infty}$-curve
that is parametrized by arclength.
Then $n\circ b : I\to \S^2$ has bounded geodesic curvature.
\end{proposition}
\begin{proof}
Clearly $\b = n\circ b$ is in $W^{1,2}$ because so is $n$.
In view of Lemma \ref{adaptedlocal} we
must show that, in the neighbourhood of every point in $\o I$,
the spherical curve $\b$ admits a $W^{1,1}$ adapted frame $\t r$ 
for which $\t r_1'\cdot\t r_2$ is bounded.
\\
For simplicity let us recall that $u\in C^1(S)$.
Moreover, by the results in \cite{HartmanNirenberg, Kirchheim, Pakzad, H-ARMA1},
for every $x_0\in S$ there exists $\delta > 0$ and a
Lipschitz continuous map $q : B_{\delta}(x_0)\to\S^1$
such that, for all $x\in B_{\delta}(x_0)$,
\begin{equation}
\label{bdcurv-1}
\D u \mbox{ is constant on the segment }B_{\delta}(x_0)\cap \left(x + \R q(x)\right);
\end{equation}
observe that this implies that $n$ is constant on these segments, too.
\\
Now let $t_0\in\o I$ and let $q$ be as above, with $x_0 = b(t_0)$.
Set $\t R_2 = q\circ b$, which is well-defined in a neighbourhood $J$ of $t_0$
in $\o I$. Defining $\t R_1 = -\t R_2^{\perp}$ and, for $i = 1, 2$,
$$
\t r_i = (D_{\t R_i}u)(b)\mbox{ on }J
$$
as well as $\t r_3 = \b$, we obtain an adapted frame $\t r$ for $\b$.
To verify that in fact $\b'$ is parallel to $\t r_1$, we notice that $\b' = D_{b'}n$,
so by symmetry of the Weingarten map we have
$$
\t r_2\cdot \b' = D_{b'}u\cdot D_{\t R_2}n = 0,
$$
in view of \eqref{bdcurv-1}.
\\
Moreover, $\t r_1'\cdot\t r_2 = \t R_1'\cdot\t R_2$
due to \eqref{h2n} and because $\D u\in O(2,3)$. And
$\t R_1'$ is bounded because $q$ is Lipschitz. Hence the adapted frame has bounded curvature.
\end{proof}

The next result asserts that, conversely to Proposition \ref{bdcurv},
given a spherical curve $\b$ with bounded geodesic curvature
and a curve $b$ in the reference domain, one can construct a local
isometric immersion whose Gauss map satisfies $n\circ b = \b$.
This allows us to pass from framed curves to ribbons with finite width.
We refer, e.g., to
\cite{FHMP-SIAM} for a related result.
The data here are different, however, as we work directly
with the Gauss map.

\begin{proposition}\label{localiso}
Let $b : \o I\to\R^2$ be a $W^{2,\infty}$ embedding parametrized by arclength
and let $\b\in W^{1,2}(I, \S^2)$ have geodesic curvature
$\kappa_g\in L^{\infty}(I)$, in the sense of Definition \ref{darb-def2}.
Set $\kappa = b''\cdot (b')^{\perp}$
and assume that
\begin{equation}
\label{loci-1}
\left|
\int_0^t (\kappa - \kappa_g)
\right| < \frac{\pi}{2}\mbox{ for all }t\in\o I.
\end{equation}
Then there exists an open neighbourhood $U$ of $b(I)$
and an isometric immersion $u\in W^{2,2}_{\delta}(U)$ satisfying 
$
n\circ b = \b\mbox{ on }I.
$
\\
More precisely,
assume that $b(0) = 0$ and $b'(0) = e_1$, set
$K_g(t) = \int_0^t\kappa_g$ and define
\begin{equation}
\label{defPhi}
\begin{split}
\Phi : I\times\R &\to\R^2
\\
(t, s) &\mapsto b(t) + s
\zwo{-\sin K_g(t)}{\cos K_g(t)}. 
\end{split}
\end{equation}
Then there exists $\e > 0$ such that $\Phi$ is invertible
on $I\times (-\e, \e)$. 
\\
Define $K(t) = \int_0^t\kappa$, let
$\t r$ be an adapted frame for $\b$ with $\t r_1'\cdot\t r_2 = \kappa_g$
and define $r : I\to SO(3)$ by setting $r_3 = \b$ and
\begin{equation}
\label{localiso-11}
r_1 = R_{\b}^{(K-K_g)}\t r_1.
\end{equation}
Then the map
$u : \Phi(I\times (-\e, \e))\to\R^3$ defined by
\begin{equation}
\label{localiso-defu}
u\left(\Phi(t,s)\right) = s\t r_2(t) + \int_0^t r_1
\mbox{ for all }(t, s)\in I\times (-\e, \e)
\end{equation}
belongs to $W^{2,2}_{\delta}\left(\Phi(I\times (-\e, \e))\right)$.
Its differential satisfies
\begin{equation}
\label{localiso-Du}
(\D u)(\Phi) = \t r_1\otimes
\zwo{\cos K_g}{\sin K_g}
+
\t r_2\otimes
\zwo{-\sin K_g}{\cos K_g}
\mbox{ on }I\times (-\e, \e)
\end{equation}
and its second fundamental form satisfies
\begin{equation}
\label{localiso-A}
A(\Phi) = \frac{-\b'\cdot\t r_1}{\cos(K - K_g) - s\kappa_g}\ \t R_1
\mbox{ on }I\times (-\e, \e).
\end{equation}
\end{proposition}
\begin{proof}
We assume without loss of generality that $b(0) = 0$ and $b'(0) = e_1$.
Define $R : I\to SO(2)$ by $R = e^{iK}$, where $K(t) = \int_0^t\kappa$.
So $R_1 = b'$.
\\
By hypothesis there exists an adapted frame
$\t r$ for $\b$ which
has geodesic curvature $\kappa_g\in L^{\infty}(I)$.
Since $\b\in W^{1,2}$, the function $\mu = -\b'\cdot\t r_1$
is in $L^2$. By Lemma \ref{adaptedunique} \eqref{adaptedunique-i} we have
\begin{equation}
\label{localiso-mu}
\b' = -\mu\t r_1.
\end{equation}
Set $\t R = e^{iK_g}$, where $K_g(t) = \int_0^t\kappa_g$.
Define $\Phi$ as in \eqref{defPhi}, i.e.,
$\Phi = b + s\t R_2$. Using $b' = R_1$ and $\t R_2'\cdot\t R_1 = -\kappa_g$, we compute
\begin{equation}
\label{localiso-au1}
\d_t\Phi = b' + s\t R_2' = \left(R_1\cdot\t R_1 - s\kappa_g\right)\t R_1 + (R_1\cdot\t R_2)\t R_2.
\end{equation}
Since
$
R_1\cdot\t R_1 = \cos\left(K - K_g\right),
$
we have
\begin{equation}
\label{loci-det}
\det\left(\d_t\Phi\ |\ \d_s\Phi \right) = \cos\left(K - K_g\right) - s\kappa_g.
\end{equation}
Since $b$ has bounded curvature, it is not hard to see that
there is an $\e > 0$ such that $\Phi$ is injective on $I\times (-\e, \e)$
and that $U = \Phi\left(I\times (-\e, \e)\right)$ is open.
We define $u : U\to\R^3$ by formula \eqref{localiso-defu}, for all $(t,s)\in I\times (-\e, \e)$.
\\
Taking the derivative with
respect to $s$ we see that
$
(D_{\t R_2}u)(\Phi) = \t r_2.
$
Using this and \eqref{localiso-au1}
as well as $r_2\cdot r_1' = \kappa$ 
(from \eqref{localiso-11}) and $R_1' = \kappa R_2$,
we find
\begin{align*}
(R_1\cdot\t R_1 - s\kappa_g)(D_{\t R_1} u)(\Phi)
+ (R_1\cdot\t R_2) \t r_2 
&= \d_t (u\circ\Phi) = r_1 - s\kappa_g\t r_1
\\
&=
\left(r_1\cdot\t r_1 - s\kappa_g\right)\ \t r_1 + (r_1\cdot\t r_2)\ \t r_2.
\end{align*}
By \eqref{localiso-11} we have $r_i\cdot\t r_j = R_i\cdot\t R_j$, so
we conclude
$
(D_{\t R_1}u)(\Phi) = \t r_1.
$
Summarizing,
\begin{equation}
\label{localiso-Du11}
(\D u)(\Phi) = \t r_1\otimes\t R_1
+
\t r_2\otimes\t R_2\mbox{ on }I\times (-\e, \e),
\end{equation}
which is \eqref{localiso-Du}.
The right-hand side of \eqref{localiso-Du}
is in $O(2,3)$, so indeed $u$ is an isometric immersion.
\\
Taking the derivative with respect to $s$ in \eqref{localiso-Du11}
we find
\begin{equation}
\label{localiso-Du1}
(D_{\t R_2}\D u)(\Phi) = 0 \mbox{ on }I\times(-\e, \e).
\end{equation}
Taking the derivative with respect to $t$ in \eqref{localiso-Du11},
using \eqref{localiso-Du1} and $\t R_2' = -\kappa_g\t R_1$
as well as $\t r_1'\cdot\b = \mu$ (which follows from \eqref{localiso-mu}),
we find
\begin{equation}
\label{localiso-Du2}
\left(R_1\cdot\t R_1 - s\kappa_g\right) (D_{\t R_1}\D u)(\Phi) =
\mu\b\mbox{ on }
I\times (-\e, \e).
\end{equation}
Now \eqref{localiso-A} follows from
\eqref{localiso-Du1}, \eqref{localiso-Du2} 
and \eqref{h2n}. And \eqref{localiso-A}
implies
\begin{equation}
\label{localiso-A2}
|A(\Phi)|^2 = \frac{|\b'|^2}{\left(\cos(K_g - K) - s\kappa_g\right)^2}
\mbox{ on }
I\times (-\e, \e).
\end{equation}
For any small $\delta > 0$, by choosing $\e > 0$ small enough
we may assume that $I\times (-\e, \e)$ is contained in
$$
M_{\delta} = \{(t, s) :
\cos\left(K(t) - K_g(t)\right) - s\kappa_g(t) \geq\delta\}.
$$
Hence
\begin{align*}
\int_U |A|^2 &\leq \int_{\Phi(M_{\delta})} |A|^2
=
\int_{M_{\delta}} |A(\Phi)|^2 |\det\D\Phi|
\\
&= \int_{M_{\delta}} \frac{|\b'|^2}{\cos (K_g - K) - s\kappa_g}
\leq\frac{1}{\delta} \int_I |\b'|^2,
\end{align*}
by definition of $M_{\delta}$. 
Since $\b\in W^{1,2}$, indeed $u\in W^{2,2}(U)$.
\end{proof}

\subsection{Framed Curves for Finite Ribbons}

We will obtain Theorem \ref{ribbon-accessible} as a consequence
of Proposition \ref{localiso} (applied with $\kappa \equiv 0$) 
and of the following result involving
only framed curves.

\begin{proposition}\label{ribaccess}
For all $\o r\in SO(3)$ and all $\o\gamma\in B_1(0)$ there 
exists 
a framed curve $(\gamma, r)$ in $C^{\infty}(\o I)$ 
with $r_1'\cdot r_2 \equiv 0$
and such that the following are satisfied:
\begin{enumerate}[(i)]
\item \label{ribaccess-i}
$r(0) = I$ and $r(1) = \o r$, as well as $\gamma(0) = 0$
and $\gamma(1) = \o\gamma$.
\item \label{ribaccess-ii}
$r_3 : I\to\S^2$ has bounded geodesic curvature.
More precisely, there exists a function $K_g\in C_0^{\infty}(I)$ with 
\begin{equation}
\label{ribaccess-0}
|K_g| \leq \frac{\pi}{4}\mbox{ on }I
\end{equation}
such that
\begin{equation}
\label{ribaccess-1}
\t r_1 = R_{r_3}^{(K_g)}r_1
\end{equation}
defines an adapted frame $\t r$ along $r_3 : I\to\S^2$.
\end{enumerate}
\end{proposition}

Proposition \ref{ribaccess} will be proven in Section \ref{Ribaccess} below.
Taking it for granted, we can prove Theorem \ref{ribbon-accessible}.

\begin{proof}[Proof of Theorem \ref{ribbon-accessible}]
Define $b : I\to\R^2$ by $b(t) = t e_1$, so that $\kappa = b''\cdot (b')^{\perp} \equiv 0$.
Let $\gamma$, $r$, $\t r$ and $K_g$ be as in Proposition \ref{ribaccess} and set $\kappa_g = K_g'$.
Then in view of \eqref{ribaccess-0} we have \eqref{loci-1}
and by \eqref{ribaccess-1}
we have \eqref{localiso-11}.
\\
Hence we can apply Proposition \ref{localiso}. We
conclude that there exists $\e > 0$ such that,
defining $\t R = e^{iK_g}$ and $\Phi(t, s) = te_1 + s\t R_2(t)$
as well as $U = \Phi(I\times (-\e, \e))$,
the map $u : U\to\R^3$ defined by
\begin{equation}
\label{ribbon-accessible-7}
u\left(\Phi(t, s)\right) = \gamma(t) + s\t r_2(t)
\end{equation}
is well-defined and
belongs to $W^{2,2}_{\delta}(U)$. And it satisfies
\begin{equation}
\label{ribbon-accessible-8}
(\D u)(\Phi) = \t r_1\otimes\t R_1 + \t r_2\otimes\t R_2.
\end{equation}
Observing that $K_g$ has compact support, we see that
$\t R = I$ and
$\t r = r$ in a neighbourhood of $\d I$. Hence $\Phi$ is the identity
in a neighbourhood of $\d I\times\R$. Moreover, \eqref{ribaccess-0}
ensures that $e_2\cdot\t R_2\geq\frac{1}{\sqrt{2}}$ on $I$.
Hence there exists $w > 0$ such that $I\times [-w, w]\subset U$.
\\
Combining \eqref{ribbon-accessible-7} and \eqref{ribbon-accessible-8}
with the boundary conditions satisfied by $\gamma$ and by $r$,
we see that $u$ and $\D u$
satisfy the asserted boundary conditions.
\end{proof}

We now turn to the proof of Proposition \ref{ribaccess}.

\subsubsection{Smooth Parametrization of Piecewise Smooth Curves}
\label{Piecewise}

Let $\t\b\in C^0([-1,1], X)$ be a curve in a smooth manifold $X$
and assume that $\t{\b}$ consists of two
smooth pieces, e.g.,
$
\t{\b}|_{[-1, 0]}\in C^{\infty}\left([-1, 0]\right)
$
and
$
\t{\b}|_{[0,1]}\in C^{\infty}\left([0,1]\right).
$
Then $\t \b$ can easily be reparametrized such that the reparametrized map is smooth.
In order to find such a smooth parametrization of $\t{\b}$, let
$\rho\in C^{\infty}(\R)$ satisfy the following conditions:
\begin{enumerate}[(i)]
\item $\rho(s) = \rho(-s)$ for all $s\in\R$.
\item $\rho > 0$ on $\R\setminus\{0\}$
\item there is an $\e\in (0, 1/4)$ such that $\rho$ is constant on $(\e, \infty)$.
\item $\rho$ itself and all of its derivatives are zero in $0$.
\item $\int_0^1\rho = 1$.
\end{enumerate}
Define $\b : [-1,1]\to\S^2$ by setting
$$
\b(t) = \t{\b}\left(\int_0^t\rho\right)
\mbox{ for all }t\in [-1,1].
$$
Then $\b\in C^{\infty}([-1, 1], X)$ is a smooth reparametrization of 
the piecewise smooth curve $\t{\b}$.

\subsubsection{Piecewise Geodesic Loop}

The following construction plays a similar role as
the one in Section \ref{Loop}.
While the curve constructed here is no longer a smooth immersion,
it satisfies additional constraints
which are essential for Proposition \ref{ribaccess}.

\begin{proposition}\label{careful loop}
Let $\a_0\in (0, \pi/4)$, let  $x$, $y\in\S^2$, 
and let $v^{(0)}\in T_x\S^2$, $v^{(1)}\in T_y\S^2$ be unit tangent vectors.
\\
Then there exists a continuous piecewise geodesic curve $\b : \o I\to \S^2$
with $\b(0) = x$ and $\b(1) = y$ 
such that the following are satisfied:
\begin{enumerate}[(i)]
\item \label{careful-i}
$\b$ is a $C^{\infty}$ immersion away from
a finite set $I'\subset I$ and
the cardinality $\# I'$
is bounded by a constant depending only on $\a_0$.
\item \label{careful-ii}
There is a function $K_g\in C_0^{\infty}(I)$
with $|K_g|\leq\a_0$ such that the frame $\t r : \o I\to SO(3)$
defined by setting $\t r_3 = \b$ and
\begin{equation}
\label{carefull-1}
\t r_1(t) = R^{(K_g(t))}_{\b(t)}\Pi_{\b|_{(0, t)}}v^{(0)}
\mbox{ for all }t\in\o I
\end{equation}
is an adapted frame for $\b$.
\item \label{careful-iii}
$\Pi_{\b}v^{(0)} = v^{(1)}$.
\end{enumerate}
\end{proposition}

\paragraph{Remarks.}
\begin{enumerate}[(i)]
\item The conclusion
of Proposition \ref{careful loop} is invariant under reparametrizations
of $\b$: if $\Psi : I\to I$ is smooth and strictly monotone, then
$\b\circ\Psi$ and $K_g\circ\Psi$
satisfy the conclusions, too.
The function $K_g\circ\Psi$ still belongs to $C_0^{\infty}(I)$.
\item In view of
Section \ref{Piecewise} we can find a $\Psi$ such that $\b\circ\Psi$
is in $C^{\infty}(\o I)$ and thus also $\t r$, defined
by \eqref{carefull-1} with $\b\circ\Psi$ instead of $\b$, is smooth.
\end{enumerate}

We begin by verifying that Proposition
\ref{careful loop} can be reduced to the particular case when
$x = y$.

\begin{lemma}\label{carefull}
Proposition \ref{careful loop} is true in the case $x = y$.
\end{lemma}

Taking this lemma for granted, we now prove Proposition \ref{careful loop} in the general case
when $x\neq y$.

\begin{proof}[Proof of Propositon \ref{careful loop}]
Let $\ell > 0$ and $\b^{(1)} : [0, \ell]\to\S^2$ be such that $\b^{(1)}$
is the shortest arclength parametrized
geodesic connecting $x$ to $y$. So $\b^{(1)}(0) = x$ and $\b^{(1)}(\ell) = y$.
(If $x$ and $y$ are antipodal then choose $\b^{(1)}$ to be
any of the two geodesics connecting them.)
\\
We apply Lemma \ref{carefull} first at $x$,
with $v^{(0)}$ as in the hypothesis of Proposition \ref{careful loop}, but with 
$$
\t v^{(1)} = (\b^{(1)})'(0)
$$
instead of $v^{(1)}$. The curve furnished by that lemma will be denoted $\b^{(0)}$.
\\
We denote by $\b^{(2)}$ the loop furnished by applying Lemma \ref{carefull}
at $y$ with $v^{(1)}$ as in the hypothesis of Proposition \ref{careful loop},
but with
$$
\t v^{(0)} = (\b^{(1)})'(\ell)
$$
instead of $v^{(0)}$.
Setting
$$
\b = \b^{(2)}\bullet \b^{(1)} \bullet \b^{(0)}
$$
we obtain the desired curve. In fact,
\begin{align*}
\Pi_{\b}v^{(0)} &=
\Pi_{\b^{(2)}}\Pi_{\b^{(1)}}(\b^{(1)})'(0) 
\\
&= \Pi_{\b^{(2)}}(\b^{(1)})'(\ell)
= v^{(1)},
\end{align*}
because $\b^{(1)} : [0, \ell]\to\S^2$ is an
arclength parametrized geodesic.
On $I^{(0)}$ and $I^{(2)}$ the function $K_g$ is
determined by Lemma \ref{carefull}, and on $I^{(1)}$ we set $K_g = 0$.
\end{proof}

\begin{proof}[Proof of Lemma \ref{carefull}]
Without loss of generality $x$ is the north pole $e_3$.
We may also assume that there exists $\theta\in (-\frac{\pi}{2}, \frac{\pi}{2}]$ such that
$$
v^{(0)} = R_{e_3}^{-\theta}e_1 = e^{-i\theta}
$$
and
\begin{equation}
v^{(1)} = R_{e_3}^{(2\theta)}v^{(0)} = e^{i\theta}.
\end{equation}
Here we identify $e^{i\delta}$ with the tangent vector
$$
\begin{pmatrix}
\cos\delta
\\
\sin\delta
\\
0
\end{pmatrix}\in T_{e_3}\S^2.
$$
The case $\theta = 0$ is trivial and the case $\theta < 0$ is obtained
by reversing the direction of travel. So we may
assume without loss of generality that $\theta > 0$.
In addition, subdividing the interval $[0, \theta]$
and iterating the following construction,
we may assume without loss of generality that 
\begin{equation}
\label{careful-theta}
\theta\in (0, 2\a_0].
\end{equation}
Geometrically, this iteration means that 
one adds several loops, where the initial direction of each loop
is dictated by a different direction $v^{(0)}$.
\\
After these reductions, we can now simply choose $\b$
to be a parametrization of two meridians.
In order to make this more precise, for $\a\in (-\pi, \pi)$
we define $b_{\a} : [0, \pi]\to\S^2$ by
$$
b_{\a}(t) =
\begin{pmatrix}
\cos\a \cdot \sin t
\\
\sin\a \cdot \sin t
\\
\cos t
\end{pmatrix}.
$$
This is an arclength parametrization of the meridian with longitude
$\a$ connecting the north pole $x = e_3$ to the south pole $-e_3$.
Now define $\b^{(1)}$, $\b^{(3)} : [0, \pi]\to\S^2$ by
$$
\b^{(1)} = b_{-\frac{\theta}{2}}
\mbox{ and }
\b^{(3)} = b_{\frac{\theta}{2}}(\pi - \cdot),
$$
and define $\b^{(0)}$, $\b^{(2)}$, $\b^{(4)} : [0, 1] \to\S^2$
by $\b^{(0)} = \b^{(4)}\equiv e_3$
and $\b^{(2)}\equiv -e_3$. We define $\b : I\to\S^2$ by
$$
\b = \b^{(4)} \bullet \b^{(3)} \bullet\b^{(2)} \bullet\b^{(1)} \bullet\b^{(0)}
$$
and the intervals $I^{(k)}\subset I$ as in Section \ref{Concatenation}
(so that, up to reparametrization, $\b$ equals $\b^{(i)}$ on $I^{(i)}$).
\\
For any $c_1$, $c_2\in\R$ and any interval $(t_0, t_1)$ let us denote by
$$
[t_0, t_1]\to\R,\ c_1\leadsto c_2
$$
some monotone function in $C^{\infty}([t_0, t_1])$
which agrees with $c_1$ near $t_0$ and with $c_2$ near $t_1$.
We define $K_g\in C_0^{\infty}(I)$ by setting
$$
K_g =
\begin{cases}
0\leadsto \frac{\theta}{2} &\mbox{ on }I^{(0)}
\\
\frac{\theta}{2} &\mbox{ on }I^{(1)}
\\
\frac{\theta}{2}\leadsto -\frac{\theta}{2} &\mbox{ on }I^{(2)}
\\
-\frac{\theta}{2} &\mbox{ on }I^{(3)}
\\
-\frac{\theta}{2}\leadsto 0 &\mbox{ on }I^{(4)}.
\end{cases}
$$
We define
$$
v(t) = \Pi_{\b|_{(0, t)}}v^{(0)} = \Pi_{\b|_{(0, t)}}e^{-i\theta}.
$$
We claim that
\begin{equation}
\label{careful-1}
v(1) = v^{(1)}
\end{equation}
and that $\t r_1(t)$ given in \eqref{carefull-1}
indeed defines an adapted frame for $\b : I\to\S^2$.
\\
In order to prove this, first notice that,
for all $\a$, $\delta\in (-\pi, \pi)$,
\begin{equation}
\label{careful-3}
\Pi_{b_{\a}|_{(0, t)}}e^{i\delta} =
R_{b_{\a}(t)}^{(\delta - \a)}b_{\a}'(t)\mbox{ for all }t\in [0, \pi].
\end{equation}
In fact, since $b_{\a}$ is an arclength parametrized geodesic,
it is enough to verify \eqref{careful-3} at $t = 0$, where
it is obvious because
$b_{\a}'(0) = e^{i\a}$.
\\
Applying \eqref{careful-3} at $t = \pi$ we obtain
\begin{align*}
\Pi_{b_{\a}}e^{i\delta} &=
R_{-e_3}^{(\delta-\a)}b_{\a}'(\pi) = R_{e_3}^{(\a - \delta)} b_{\a}'(\pi)
= - e^{i(2\a - \delta)}.
\end{align*}
Thus
$$
\Pi_{b_{\theta/2}}v^{(0)} 
= -e_1
= \Pi_{b_{-\theta/2}} v^{(1)}.
$$
This proves \eqref{careful-1}, which in turn implies \eqref{careful-iii}.
\\
Applying \eqref{careful-3} with $\delta = -\theta$ and $\a = -\theta/2$
we have
$$
\Pi_{b_{\frac{\theta}{2}}|_{(0,t)}}v^{(0)}
= R_{b_{-\frac{\theta}{2}}(t)}^{(-\theta/2)}b_{-\frac{\theta}{2}}'(t)
\mbox{ for all }t\in [0, \pi].
$$
So
\begin{equation}
\label{hispa-1}
\b' = |\b'| R_{\b}^{(\theta/2)}v \mbox{ on }I^{(1)}
\end{equation}
because $\b$ is a constant speed reparametrization of 
$b_{-\frac{\theta}{2}}$ on $I^{(1)}$.
Similarly, \eqref{careful-3} shows that
\begin{equation}
\label{hispa-2}
\b' = -|\b'| R_{\b}^{(-\theta/2)}v \mbox{ on }I^{(3)}
\end{equation}
because $\b$ is a constant speed reparametrization of 
$b_{\frac{\theta}{2}}(\pi - \cdot)$ on $I^{(3)}$.
\\
We define $\t r_1$ by \eqref{carefull-1}, that is,
$$
\t r_1 = R_{\b}^{(K_g)}v.
$$
In view of \eqref{hispa-1}, \eqref{hispa-2}
and since 
$$
\b' = 0\mbox{ on }I^{(0)}\cup I^{(2)}\cup I^{(4)},
$$
we see that $\t r_1\times\b' = 0$ everywhere.
Hence $\t r$ is an adapted frame for $\b$.
In view of \eqref{careful-theta} 
this implies \eqref{careful-ii}.
\\
Finally, \eqref{careful-i} is an immediate consequence of the construction.
\end{proof}

\subsubsection{Proof of Proposition \ref{ribaccess}}\label{Ribaccess}

As in the proof of Proposition \ref{accessible}
we choose $N\in\N$ and $\o r_1^{(1)}, ..., \o r_1^{(N)}\in\S^2$ such that
$\o\gamma$ lies in the interior of the convex hull of $\{\o r_1^{(1)}, ..., \o r_1^{(N)}\}$;
we set $\o r^{(0)}_1 = e_1$ and $\o r_1^{(N+1)} = \o r_1$, and we
choose a partition $0 = t_0 < t_1 < \cdots < t_N < t_{N+1} = 1$ of $I$.
Let $\o\b$ be the shortest arclength parametrized geodesic
with $\o\b(0) = e_3$ and $\o\b(1) = \o r_3$.

For $i = 0, \dots, N$ denote by $\b^{(i)} : I\to\S^2$ the curve
obtained by applying Proposition \ref{careful loop} with
$x = \o\b(t_i)$ and $y = \o\b(t_{i+1})$ as well as
$$
v^{(0)} = \o r_1^{(i)}
\mbox{ and }
v^{(1)} = \o r_1^{(i+1)}.
$$
Define $\b = \b^{(N+1)}\bullet \cdots\bullet\b^{(0)}$.
Define $r : I\to SO(3)$ by setting
$
r_3 = \b
$
and
\begin{equation}
\label{addingloop-1}
r_1(t) = \Pi_{\b_{(0, t)}} e_1.
\end{equation}
Then by Proposition \ref{careful loop} there is a 
function $K_g\in C_0^{\infty}(I)$ with \eqref{ribaccess-0}
such that, defining $\t r_1$ by \eqref{ribaccess-1},
we have
\begin{equation}
\label{ribaccess-2}
\b'\times \t r_1 = 0\mbox{ on }I.
\end{equation}
By Proposition \ref{careful loop} 
the curve $\b$ is continuous on $I$, as well as
smooth and immersed away from a finite set $I'\subset I$.
Hence $r$ is smooth away from $I'$, too.
By the remark following Proposition \ref{careful loop},
after a suitable simultaneous reparametrization of $r$ and $K_g$
we may assume that $r$ is in $C^{\infty}(\o I)$ and that $\b'\neq 0$
away from a finite number of points. Then the frame $\t r$ defined
by $\t r_3 = \b$ and by \eqref{ribaccess-1} is in $C^{\infty}(\o I)$.

Now we reparametrize again, this time as described in Section
\ref{Convex Integration}. Then the curve $\gamma : I\to\R^3$
defined by $\gamma(t) = \int_0^t r_1$ satisfies $\gamma(1) = \o\gamma$.
As before, the function $K_g$ (hence $\t r$) is reparametrized as well.
This yields a framed curve $(\gamma, r)$ and a frame $\t r$ adapted
to $r_3 : I\to\S^2$ with the desired properties.

\bigskip

\paragraph{Acknowledgements.} This work was supported by DFG
SPP 2256.

\def\cprime{$'$}

\bibliographystyle{acm}

\end{document}